\documentclass[%
reprint,
superscriptaddress,
amsmath,amssymb,
aps,
prl,
]{revtex4-1}

\usepackage{graphicx}
\usepackage{xcolor}
\usepackage[normalem]{ulem}
\usepackage{dcolumn}
\usepackage{bm}
\usepackage{mathtools}
\usepackage[thinc]{esdiff}
\usepackage{textcomp}
\usepackage{bbold}
\usepackage{braket}
\usepackage{hyperref}
\usepackage{amsthm}
\newtheorem{theorem}{Theorem}
\newtheorem{lemma}[theorem]{Lemma}

\begin{document}


\title{Quantum spatial search in one dimension through long-range interactions}
\title{Optimal quantum spatial search with one-dimensional long-range interactions}

\author{Dylan Lewis}
 \email{dylan.lewis.19@ucl.ac.uk}
\affiliation{
 Department of Physics and Astronomy, University College London, 
 London WC1E 6BT, United Kingdom
}
\author{Asmae Benhemou}
\affiliation{
 Department of Physics and Astronomy, University College London, 
 London WC1E 6BT, United Kingdom
}
\author{Natasha Feinstein}
\affiliation{
 Department of Physics and Astronomy, University College London, 
 London WC1E 6BT, United Kingdom
}
\author{Leonardo Banchi}
\affiliation{Department of Physics and Astronomy, University of Florence, 
via G. Sansone 1, I-50019 Sesto Fiorentino (FI), Italy}
\affiliation{INFN Sezione di Firenze,  via G. Sansone 1,  I-50019 Sesto Fiorentino (FI), Italy}
\author{Sougato Bose}
\affiliation{
 Department of Physics and Astronomy, University College London, 
 London WC1E 6BT, United Kingdom
}

\date{\today}

\begin{abstract}
Continuous-time quantum walks can be used to solve the spatial search problem,
which is an essential component for many quantum algorithms that run
quadratically faster than their classical counterpart, in $\mathcal O(\sqrt n)$  time
for $n$ entries. However, the capability of models found in nature is largely unexplored -- e.g., in one dimension only nearest-neighbour Hamiltonians have been considered so far, for which 
the quadratic speedup does not exist. 
Here, we prove that optimal spatial search, namely with $\mathcal O(\sqrt n)$ 
run time and high fidelity, is possible in one-dimensional spin
chains with long-range interactions that decay as $1/r^\alpha$ with distance $r$.
In particular, near unit fidelity is achieved for $\alpha\approx 1$ and,
in the limit $n\to\infty$, we find a continuous transition from a region where 
optimal spatial search does exist ($\alpha<1.5$) to where it does not ($\alpha>1.5$).
Numerically, we show that spatial search is robust to dephasing
noise and that, for reasonable chain lengths,
$\alpha \lesssim 1.2$ should be sufficient to demonstrate
optimal spatial search experimentally with near unit fidelity.
\end{abstract}

\maketitle

\paragraph*{\label{sec:intro}Introduction.}
Spatial search is the problem of finding a marked element in a graph with $n$ nodes. For classical algorithms, there is no shortcut and $\mathcal{O}(n)$ queries are required. However, with quantum algorithms, spatial search can be solved optimally in $\mathcal{O}(\sqrt{n})$ time~\cite{Grover1996ASearch,Bennett1997StrengthsComputing}. Childs and Goldstone~\cite{Childs2004} found that spatial search can be solved by an algorithm using a continuous-time quantum walk. They showed that for the complete graph, the hypercube graph, and $d$-dimensional periodic lattices of $d>4$, the marked node can be found in optimal time. Since then, a number of graphs have been found that permit optimal spatial search~\cite{Novo2015SystematicGraphs,Chakraborty2016SpatialGraphs,Chakraborty2017OptimalNetworks,Novo2018Environment-assistedSearch, Wong2018QuantumGraphs,Osada2020Continuous-timeNetwork, Sato2020ScalingWalk}. Recently, Chakraborty et al.~\cite{Chakraborty2020OptimalityWalks} have shown necessary and sufficient conditions for optimal spatial search for graphs with a sufficiently large spectral gap.

A quantum walk is the quantum equivalent of the classical random walk. Unlike a classical random walk, the quantum walker takes a superposition of paths~\cite{Kempe2003QuantumOverview}. The interference between those paths forms the basis of the quantum algorithms that utilise quantum walks.
From a graph-theoretical perspective, 
a continuous-time quantum walk is generated by a unitary evolution defined by the adjacency matrix $A$ of a graph~\cite{Farhi1998QuantumTrees,Childs2003ExponentialWalk}. The vertices of the graph define orthonormal basis states of a Hilbert space and an evolution for time $t$ is given by $e^{-iAt}$. The latter is equivalent to the natural dynamics of a quantum system where the Hamiltonian is the adjacency matrix of the graph that defines the hopping between basis states. 

Can Hamiltonians found in nature, with interactions typically falling off with distance, admit a quantum walk capable of optimal spatial search? This is an important question in the Noisy Intermediate-Scale Quantum (NISQ) \cite{Preskill2018} era where we can look for the quantum speedup capabilities of non-error-corrected collections of qubits with physically motivated couplings. Most of the previously studied Hamiltonians that admit optimal spatial search are difficult to find in
nature and must be artificially synthesized, e.g. using quantum simulation
techniques~\cite{Georgescu2014QuantumSimulation} or a digital quantum computer.  Even a recent idea of using an unweighted long-range percolation graph~\cite{Osada2018SpatialInteractions} would amount to having stochastic interactions, whose realization is unclear.   
Childs and Ge~\cite{Childs2014SpatialLattices} also noted that an interaction strength that decays as a quadratic power law with distance would be sufficient to find optimal spatial search in $d=2$. 
For one-dimensional systems, 
only nearest-neighbour interactions have been considered, where optimal quantum 
search was shown to be impossible~\cite{Childs2004}. 

Here, we propose a physically-motivated model for spatial search on a closed one-dimensional spin chain using long-range interactions that decay as $1/r^\alpha$, with $r=|i-j|$ the distance between lattice sites $i$ and $j$, and $\alpha \ge 0$. At the moment, this is a highly topical model, realizable in ion traps \cite{Kim2009EntanglementModes, Lanyon2017EfficientSystem, Richerme2014Non-localInteractions, Friis2018ObservationSystem}, dipolar crystals~\cite{Micheli2006AMolecules,Yan2013ObservationMolecules}, Rydberg arrays~\cite{Browaeys2020Many-bodyAtoms} etc., although the tunability of $\alpha$ is probably present only in the ion trap setting. This model has been explored for its capabilities of scrambling \cite{Chen2019FiniteInteractions}, novel dynamical quantum phase transitions \cite{Zhang2017ObservationSimulator}, and quantum state transfer \cite{Eldredge2017FastInteractions,Tran2020OptimalSystems}. Yet its potential for quantum computation remains unexplored.  

In the case $\alpha=0$, we have the complete graph. This is equivalent to nearest-neighbour interactions in a spatial dimension equal to the number of spins $n$. For complete graphs, optimal spatial search has been shown~\cite{Farhi1998}. In the case of large $\alpha$, the graph approaches the one-dimensional periodic lattice, where spatial search does not exist. This picture, where long-range interactions effectively mask the dimension of the system~\cite{Maghrebi2015ContinuousSystems}, suggests there is a transition between the regime where optimal spatial search exists and where it does not. Here, we address the following questions: for which values of $\alpha$ can we show optimal search, and at what value of $\alpha$ is there a transition between the two regimes? 
We will show both numerically and analytically for large $n$ that optimal spatial search, namely with $\mathcal O(\sqrt n) $
run time, does exist for $\alpha<1.5$ and has a near perfect fidelity  for $\alpha\lesssim 1.2$.  We note that the interaction strengths found for optimal search are experimentally realizable. In particular, ion trap experiments have been performed for chains of ions with interaction strengths with a potential of $\alpha \approx 1$~\cite{Kim2009EntanglementModes, Lanyon2017EfficientSystem,Richerme2014Non-localInteractions, Friis2018ObservationSystem}. In principle, $\alpha$ can be tuned to anywhere between 0 and 3 for low $n$~\cite{Richerme2014Non-localInteractions}. However, as $n$ increases, $\alpha \ll 1$ becomes experimentally more difficult. We show that experimental designs using ion traps along these lines would be able to demonstrate optimal spatial search.

\paragraph*{\label{sec:spatial_search}Spatial search.}
A quantum search problem can be solved in $\mathcal{O}(\sqrt{n})$ time using Grover's algorithm~\cite{Grover1996ASearch}. An analog version of this search algorithm was suggested by Farhi and Gutmann~\cite{Farhi1998}, which is a continuous-time quantum walk on a complete graph~\cite{Childs2004}. It is experimentally difficult to encode this search in the general case for the entire Hilbert space, which gives a graph of size $2^n$ for $n$ spins, because it is hard to implement the continuous-time oracle Hamiltonian and the graph Hamiltonian. We therefore restrict ourselves to the spatial search problem in the single-excitation basis of size $n$, which can naturally be mapped to a physical setting. Each of the $n$ vertices of the search graph, $G$, represent a single excited spin. The basis states of this space are therefore $\ket{j} = \ket{0}_1 \otimes \dots \otimes \ket{0}_{j-1} \otimes \ket{1}_j \otimes \ket{0}_{j+1} \otimes \dots \otimes \ket{0}_n = \ket{0...010...0}$. The marked state is identified by measuring the system to locate the excited spin. The oracle Hamiltonian is simply a local magnetic field at the marked spin site, 
\begin{equation}
    H_{\textrm{marked}} = |w \rangle\langle w|,
\end{equation}
where $w\in\{1,\dots,n\}$ labels the marked vertex of the graph. The search includes the graph Hamiltonian and the marked state Hamiltonian,
\begin{equation}
    H_{\textrm{search}} = \gamma H + H_{\textrm{marked}},
\end{equation}
where the relative strength of the two Hamiltonians is given by $\gamma$, an effective hopping rate for the quantum walker between vertices of the graph. We use the original Childs and Goldstone spatial search algorithm for continuous-time quantum walks. First, a specific value of $\gamma$ is chosen. Then the system is initialized in a specific state $|s\rangle$ and evolved under the system dynamics for a time $T$. Finally, the state is measured and the marked state is found with probability $F$
\begin{equation}
		\label{F}
    F = \left| \langle w | e^{-iH_\textrm{search}T} | s \rangle \right|^2.
\end{equation}
The aim is to find an $F$ as close to 1 as possible.
The search is optimal if $T$ is $\mathcal{O}(\sqrt{n})$. The initial state is 
\begin{equation}
    |s\rangle = \frac{1}{\sqrt{n}}\sum_{j=1}^{n}|j\rangle,
\end{equation}
which has an overlap of $1/\sqrt{n}$ with the marked state. 

For spatial search, we want the system to oscillate between the states $\ket{s}$ and $\ket{w}$ with high fidelity~\cite{Byrnes2018GeneralizedStates, Cafaro2019Continuous-timeSystems}. The $\gamma$ that maximises the overlap of the dominant eigenvectors of the system with the states $\ket{s}$ and $\ket{w}$ achieves this~\cite{Childs2004}. The maximum overlap occurs at the minimum gap, that is the minimum energy difference between the ground state energy and the first excited state energy. The time it takes to oscillate between the superposition state and the marked state is proportional to the inverse of the energy gap. The minimum energy gap is proportional to $1/\sqrt{n}$, and therefore the time to reach the marked state from the superposition state is proportional to $\sqrt{n}$. This gives optimal spatial search and explains how the optimum $\gamma$ for the system can be found.

\paragraph*{\label{sec:model}Long-range interaction model.}
Long range interactions can be realized in one-dimensional ion-trap systems, where spin-spin couplings are generated through laser-induced forces that off-resonantly drive vibrational modes of the ion chain~\cite{Richerme2014Non-localInteractions, Friis2018ObservationSystem}. This interaction is well described by the XY model with Hamiltonian~\cite{Kim2009EntanglementModes, Lanyon2017EfficientSystem},
\begin{equation}
    H = \sum_{i<j} J_{ij} \left( \sigma^x_{i}\sigma^x_{j} + \sigma^y_{i}\sigma^y_{j} \right),
\end{equation}
where the interaction strength, $J_{ij}$, is dependent on the distance between the spins,
\begin{equation}
    J_{ij} =  \frac{1}{|j-i|^\alpha} + \frac{1}{|n-(j-i)|^\alpha},
\end{equation}
considering a closed one-dimensional spin chain.

\begin{figure}[ht]
    \centering
		\includegraphics[width=0.48\textwidth]{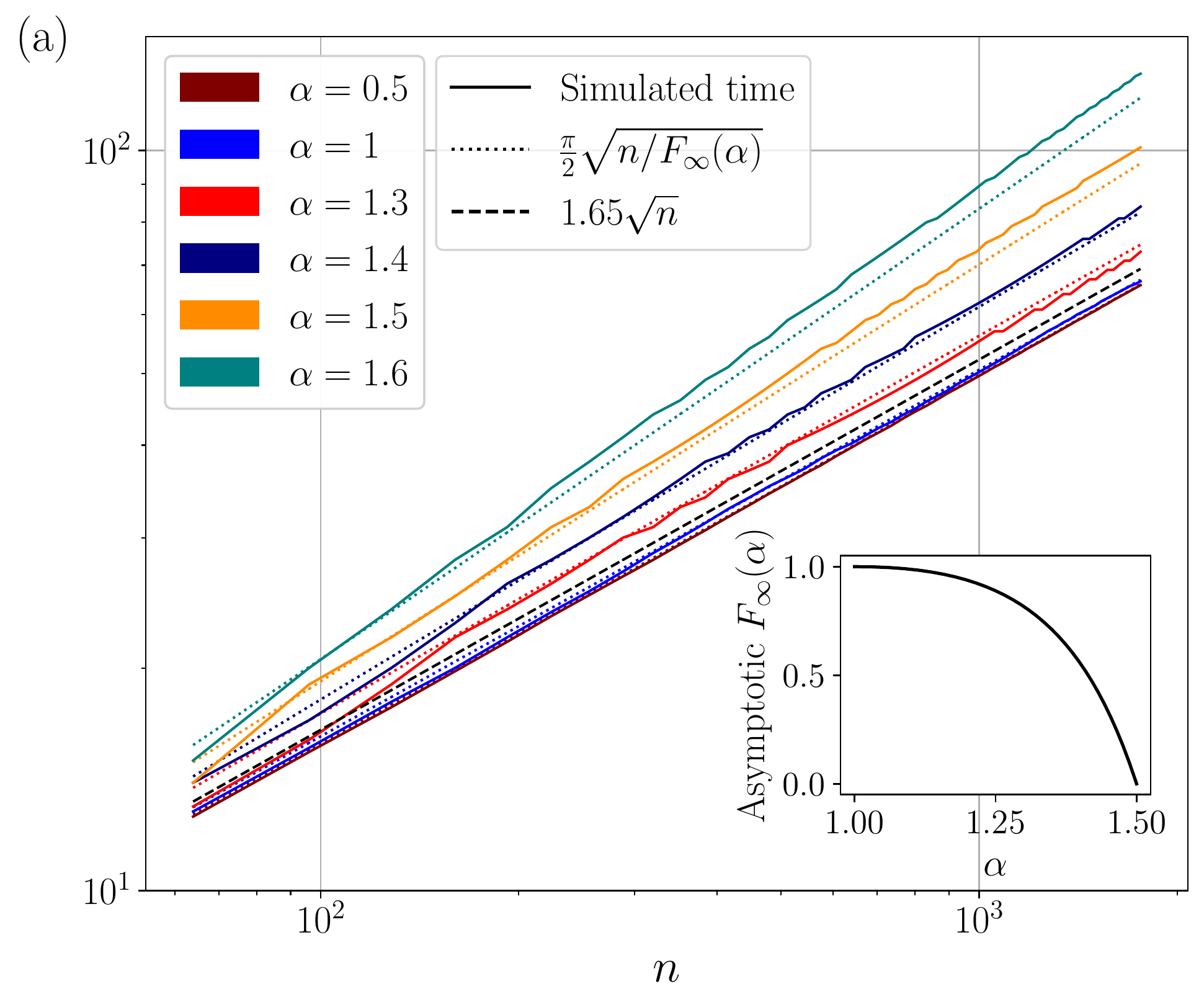}
    \includegraphics[width=0.48\textwidth]{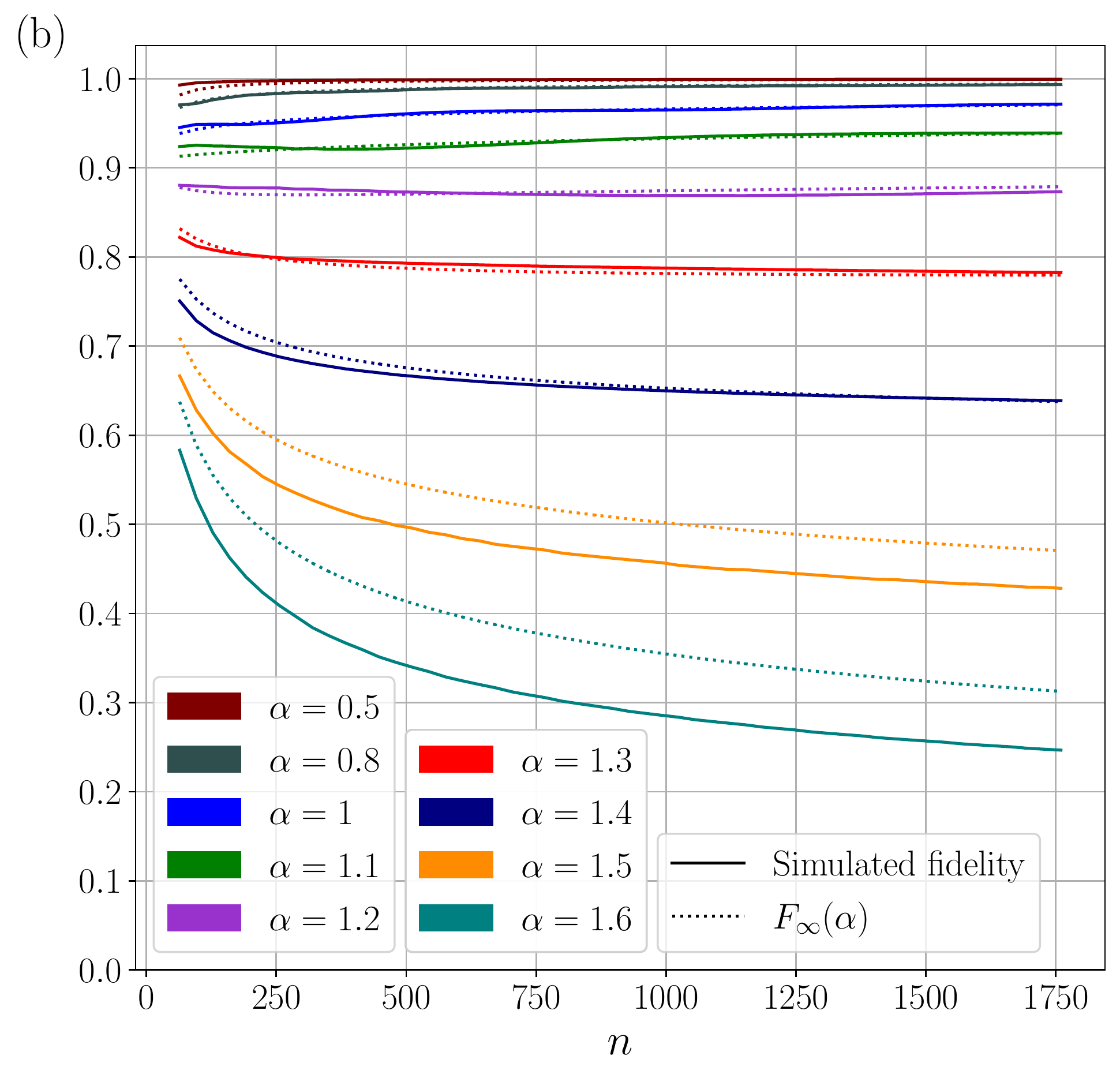}
		\caption{ 
		(a) Time to reach the maximum fidelity of the marked state against the number of spins $n$.
	  (b) Maximum fidelity of the spatial search against $n$.
		Numerical results are compared with the analytical predictions:
		$\frac \pi2\sqrt{n / F_\infty(\alpha)}$ for time, and $F_\infty(\alpha)$ for the fidelity, as defined in Eq.~\eqref{F}.
        The line (dashed black) $1.65\sqrt{n}$ is plotted in (a) as a reference for how well the time scales as $\mathcal{O}(\sqrt{n})$.
		Inset: asymptotic value of $F_\infty(\alpha)$, in the limit $n \rightarrow \infty$, as a function of $\alpha$ using the approximation of $F_\infty(\alpha)$ from Eq.~\eqref{eq:approx_nu}. (Supplementary Material contains an additional plot for time with fit parameters and a plot for fidelity that shows the asymptotic limits.)
	}
	\label{fig:timefidelity}
\end{figure}

In order to demonstrate optimal spatial search for this model, we must show that the time to reach the maximum fidelity scales $\sim \sqrt{n}$. Fig.~\ref{fig:timefidelity}(a) compares numerically simulated time to reach the maximum fidelity with an analytical prediction $\frac{\pi}{2}\sqrt{n/F_\infty(\alpha)}$, where we have analytically found an approximation of the fidelity for large n, $F_\infty(\alpha)$ (see the Supplementary Material).
We observe that the search time from numerical simulations closely follows the analytical times for $\alpha<1.5$. This implies spatial search in optimal $\mathcal{O}(\sqrt{n})$ time. 
Furthermore,
Fig.~\ref{fig:timefidelity}(b) shows that the fidelity for $\alpha < 1.5$ is closely approximated by $F_\infty(\alpha)$. 
The asymptotic value of $F_\infty(\alpha)$ is displayed in the inset of Fig.~\ref{fig:timefidelity}(a). Thus, whenever $F_\infty(\alpha)$ is asymptotically non-zero, there is optimal spatial search in $\mathcal{O}(\sqrt{n})$ time. We also note that for $\alpha\leq 1.1$, the fidelity is high (above 0.9) for low $n$ and approaches its asymptotic value relatively quickly. Even before the asymptotic value is reached, the scaling with $n$ is not significant and essentially gives optimal spatial search for all $n$. We have confirmed that these results are irrespective of the marked state chosen. 
In particular, the fidelity is high for $\alpha=1$ and reaches over 0.97 for 1760 spins, while for 
$\alpha = 1.4$ the fidelity is low, around 0.6 for $n\sim 3\times10^6$ (not shown). 

\paragraph*{\label{sec:transition}Optimal search regime.}
We apply the criterion for the optimality of quantum search from Ref.~\cite{Chakraborty2020OptimalityWalks},  
to investigate the $\alpha$ values that permit optimal spatial search. This criterion is valid when the spectral condition
\begin{equation}
    \label{eq:spectral_condition}
    \Delta \ge c n^{-\frac{1}{2}}
\end{equation}
is satisfied, where $c$ is a small positive constant, and $\Delta(\alpha) =  1 - \Tilde{\lambda}_{n-1}(\alpha)$ is a rescaled spectral gap -- $\Tilde{\lambda}_k$ are the eigenvalues of $H$ in increasing order, rescaled such that the largest eigenvalue is 1 and the smallest is 0. 
When the spectral condition is satisfied, for large $n$, the fidelity $F_\infty$ and time $T$ are related by $T \approx \frac{\pi}{2}\sqrt{n/F_\infty}$. 
Therefore, if $F_\infty(\alpha)\to 0$ asymptotically, optimal spatial search in $\mathcal O(\sqrt n)$ time is not possible.

The region where the spectral condition applies, as determined by Eq.~\eqref{eq:spectral_condition}, is dependent on $\alpha$. Asymptotically with respect to $n$, we find the spectral gap for $\alpha<3$ and $\alpha \ne 1$,
\begin{equation}
    \label{eq:gap_asymptotic}
    \Delta(\alpha) \sim 1 -  \frac{1 - \frac{g_0(\alpha)}{f(\alpha)}n^{1-\alpha}}{1 - \frac{2}{f(\alpha)} \frac{n^{1-\alpha}}{\alpha-1}},
\end{equation}
where $f(\alpha) =  4\zeta(\alpha) -2^{2-\alpha}\zeta(\alpha)$, $g_0(\alpha)  = -2^{\alpha}\pi^{\alpha-1}\sin(\tfrac{\alpha\pi}{2})\Gamma(1-\alpha)$, $\zeta(\alpha)$ is the Riemann zeta function, and $\Gamma(1-\alpha)$ is the gamma function. This is proved in the Supplementary Material. 
For $\alpha<1$, we therefore find $\Delta(\alpha) = \mathcal{O}(1)$, which satisfies the spectral condition. Using Lemma 5 from Ref.~\cite{Chakraborty2020OptimalityWalks}, $\Delta(\alpha)(1-\frac{1}{n}) \leq F_\infty(\alpha) \leq 1$. Therefore, $F_\infty(\alpha)$ must tend to a constant. This proves optimal spatial search exists for $\alpha < 1$. Numerically this fidelity tends to 1, see Fig.~\ref{fig:timefidelity}(b).

For $\alpha  > 1.5$, optimal spatial search cannot exist
due to the Lieb-Robinson bounds for long-range interactions~\cite{Lieb1972TheSystems, Tran2018LocalityInteractions, Kuwahara2020StrictlyDimensions}. These bounds give an effective light-cone
for the maximum correlation distance $r$ after time $t$.
In a one-dimensional free-particle system $t=\mathcal{O}(r^{\alpha-1})$ 
for $1<\alpha<2$, while $t=\mathcal{O}(r)$ for $\alpha\geq 2$~\cite{Tran2020HierarchyInteractions}.  
For spatial search, since the maximum distance is $r=\mathcal{O}(n)$,
we find that the time must be lower-bounded by 
$t=\mathcal{O}(n^{\alpha-1})$, thus showing that optimal spatial search 
with $t=\mathcal{O}(\sqrt n)$ is not possible for $\alpha>1.5$.

For $1<\alpha<1.5$, we find $\Delta(\alpha)=\mathcal{O}(n^{1-\alpha})$.
Combining this scaling with  Eq.~\eqref{eq:spectral_condition},
we have therefore proved that the spectral condition is asymptotically satisfied for $\alpha < 1.5$.  This is demonstrated numerically in Fig.~\ref{fig:spectral_gap_scaling}. 
Moreover, in the Supplementary Material, we find the asymptotic expansion
\begin{equation}
    \label{eq:approx_nu}
    F_\infty(\alpha) \approx \frac{2 \left(n^{\alpha-2} \zeta(\alpha -1) + \frac{2^{\alpha -2}}{2-\alpha}\right)^2 }{{n^{2\alpha-3} \zeta(2\alpha -2) + \frac{2^{2\alpha -3}}{3-2\alpha}}},
\end{equation}
and prove that $F_\infty(\alpha)$ also approaches a non-zero value for $1 < \alpha<1.5$. Therefore, optimal spatial search exists for $\alpha<1.5$, although
for $\alpha$ close to 1.5 the fidelity converges to a low value, despite being optimal.
\begin{figure}[t]
    \centering
		\includegraphics[width=0.48\textwidth]{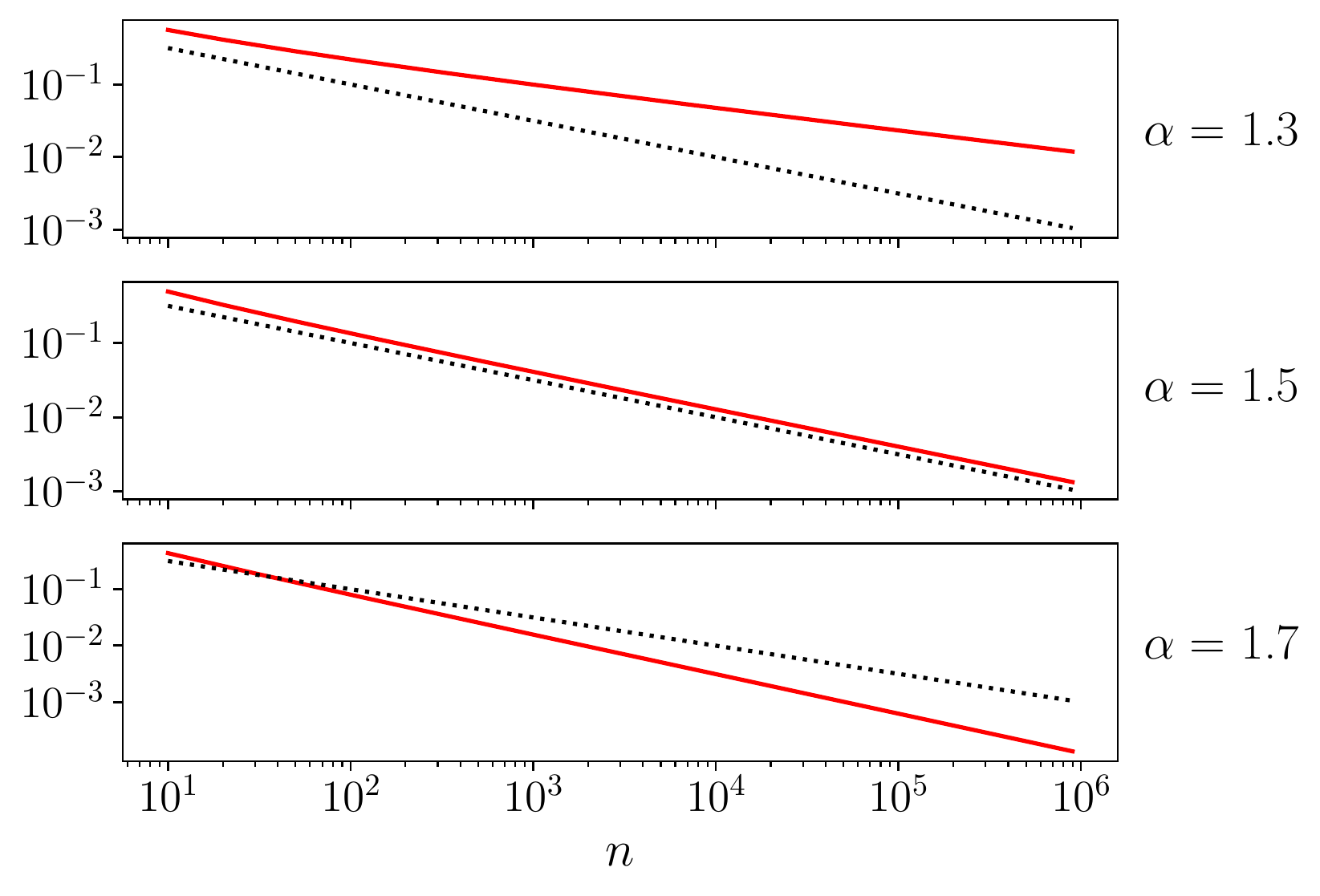}
    \caption{The spectral condition is shown by comparing $\Delta(\alpha)$ from Eq.~\eqref{eq:gap_asymptotic} (solid red) with $n^{-\frac{1}{2}}$ (dotted black) for various $\alpha$. 
	The small positive constant $c$ from Eq.~\eqref{eq:spectral_condition} has been ignored, as it provides a constant shift. 
	}
    \label{fig:spectral_gap_scaling}
\end{figure}
The inset of Fig.~\ref{fig:timefidelity}(a) illustrates the continuous transition between the two regimes where, 
asymptotically, optimal spatial search exists ($\alpha<1.5$) or is impossible ($\alpha>1.5$). 
This curve is reminiscent of the behaviour of the order parameter when the controls pass a phase transition 
point, although with a different physical explanation. 
The asymptotic fidelity $F_\infty(\alpha)$ predicts perfect search for $\alpha=1$, with a significant
decrease after
$\alpha=1.3$, before reaching 0 at $\alpha = 1.5$,
see the Supplementary Material for
details.

\paragraph*{\label{sec:noise}Dephasing noise.}
In a physical implementation, the dephasing of the qubits is the principal impediment. A pure dephasing which allows coherence to be lost without energy exchange with the environment, is the most dominant in ion traps, for example \cite{Piltz2016VersatileProcessing}. This can be modeled as random local field fluctuations by adding a noise term to the diagonal elements of the system Hamiltonian. The noise has a mean of 0 and is sampled from a Gaussian distribution, the standard deviation of which defines the noise parameter. Our results were obtained by generating $100$ such Hamiltonians, running the spatial search for each one, and averaging over the outputs. 

Fig. \ref{fig:dephasing_comp} shows the evolution of the quantum walk on a closed spin chain of $256$ spins with $\alpha=1$ at four different levels of noise. We find that the quantum walk on the ring is reasonably robust against dephasing and the maximum fidelity significantly falls when the noise level is greater than $0.02$. This corresponds to the field fluctuations being on the order of $2\%$ of the field used to differentiate the marked state. The time to reach the maximum fidelity is also only minimally decreased at this noise level. The robustness of a quantum walk on a spin chain to small fluctuations in the interactions between sites has been shown analytically in the supplementary material of Ref.~\cite{Pitsios2017PhotonicQuench}.

\begin{figure}
    \centering
    \includegraphics[scale=0.48]{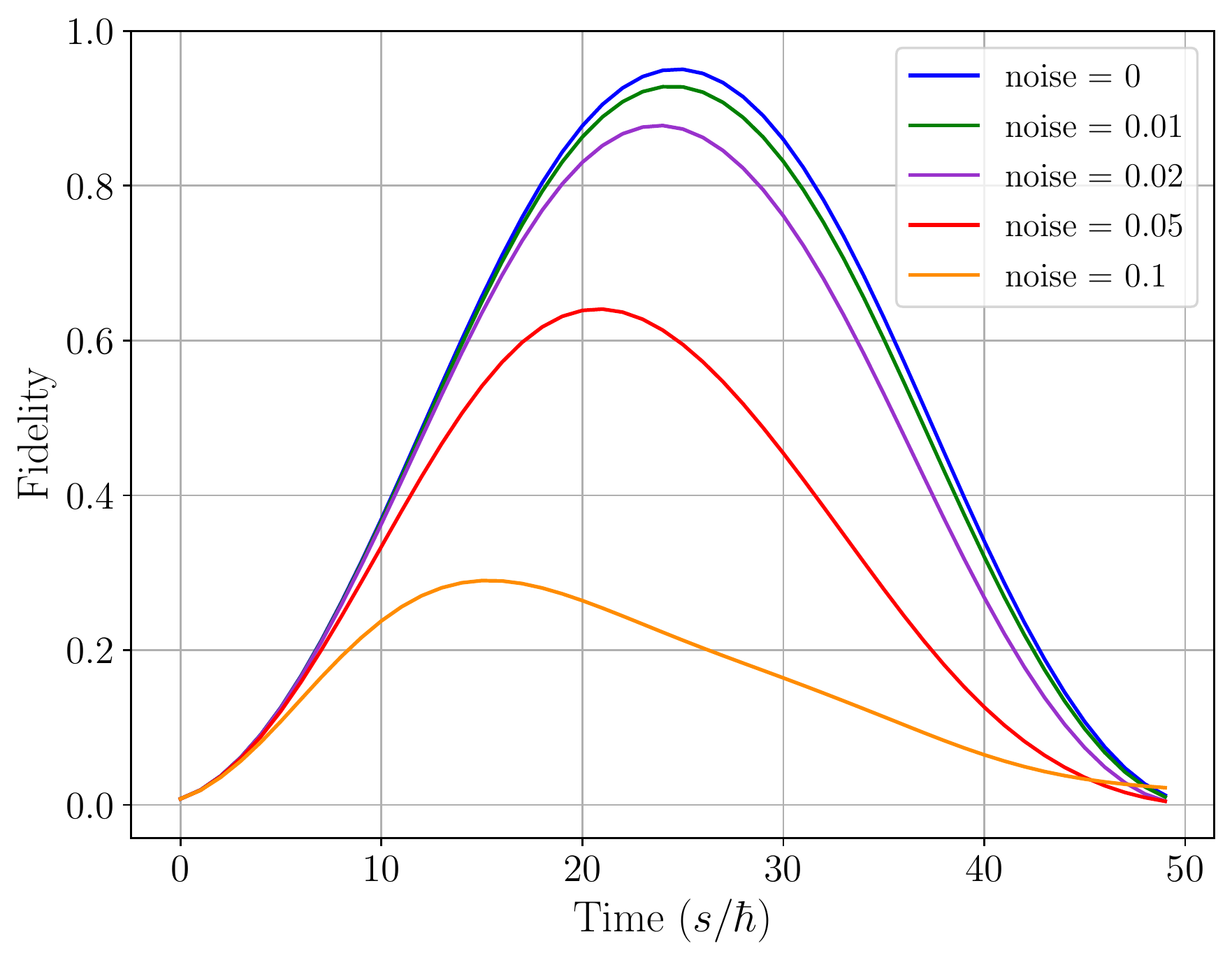}
    \caption{Fidelity of the marked state over time for a closed spin chain of 256 spins with $\alpha=1$. Noise is the standard deviation of a Gaussian distribution for the local field applied to each site. The noise can be compared to the marked site field which has a magnitude of 1. }
    \label{fig:dephasing_comp}
\end{figure}

\paragraph*{\label{sec:conslusion}Discussion.}
We have shown the possibility of a quantum speedup in a NISQ device with permanent long-range interactions. It can also be regarded as a quantum computation application of a quantum simulator -- in this sense, it is even less demanding than what is envisaged in a typical NISQ device because it does not require individual gates between distinct qubits, it merely requires a time-independent many-body Hamiltonian to be switched on, and then switched off after a specified interval of time. Our result is that optimal spatial search is physically realizable in one-dimensional spin chains through long-range interactions that decay as a $1/r^\alpha$ potential. We have demonstrated analytically and confirmed numerically that optimal spatial search in $\mathcal{O}(\sqrt{n})$ time exists for $\alpha < 1.5$. As $\alpha$ approaches 1.5, however, the fidelity becomes impractically low. For experimentally realistic $n$, around 50 to 100, the interaction range $\alpha \lesssim 1.2$ gives a fidelity above $0.88$. 
We have also shown that dephasing noise of $\sim 1\%$ of the marking field only slightly reduces fidelity. Therefore, without considering specific implementations, we argue that optimal spatial search could be achieved experimentally for $\alpha \lesssim 1.2$. Long-range interactions with these values of $\alpha$ have been demonstrated in ion traps~\cite{Richerme2014Non-localInteractions, Lanyon2017EfficientSystem, Friis2018, Pagano2019QuantumSimulator}. Using key results from Chakraborty et al.~\cite{Chakraborty2020OptimalityWalks}, we were able to show that an interaction strength of $\alpha = 1.5$ defines a phase transition-like point for optimal spatial search, where the asymptotic fidelity behaves like an order parameter. 

Although we have shown optimal spatial search for the ring geometry of a closed one-dimensional chain, optimal scaling also exists for open spin chains, with the time to reach marked states at the edge of the chain being longer than central marked states. We can motivate the investigation of the ring as equivalent to the central spins of a much longer open chain. We also note that optimal spatial search should exist for the more connected 2-dimensional periodic lattice with long-range interactions, likely with a transition at a higher value of $\alpha$. 

From a physical point of view, the spatial search algorithm allows detection of a local magnetic field faster than is possible classically. It could therefore find use in quantum sensing as a protocol for locating short-lived local magnetic fields along a spin chain -- or extended to a 2-dimensional lattice. Perhaps such a device could be used for recognizing features in image processing faster than classically possible if the image is encoded as an array of fields that mark different qubits. By comparison with the theoretical maximum for the fidelity, the spatial search algorithm could also be used to establish how well coupled the long-range interactions of a spin chain are.

\begin{acknowledgements} \paragraph*{Acknowledgements.}
D.L., A.B., and N.F. acknowledge support from the EPSRC Centre for Doctoral Training in Delivering Quantum Technologies, grant ref. EP/S021582/1. L.B.  acknowledges  support  by  the  program ``Rita  Levi  Montalcini''  for  young  researchers. 
\end{acknowledgements}

\clearpage

\onecolumngrid
\appendix
\section*{Supplementary Material}
\section{\label{sec:asymptotic_eigenvalues}Hamiltonian properties and asymptotic expansion of the rescaled eigenvalues}
The Hamiltonian for the XY model can be written in the single-excitation basis as 
\begin{equation}
    H = \sum_{i<j} J_{ij} \left(|i\rangle\langle j| + |j\rangle\langle i| \right).
\end{equation}
This gives a symmetric circulant matrix, where the interaction strength, $J_{ij}$, is dependent on the distance between the spins,
\begin{equation}
    J_{ij} =  \frac{1}{|j-i|^\alpha} + \frac{1}{|n-(j-i)|^\alpha},
\end{equation} 
with an interaction strength determined by $\alpha$. The eigenstates of $H$ are the Fourier states,
\begin{equation}
    |\phi(k)\rangle = \frac{1}{\sqrt{n}}\sum_{j=1}^{n}e^{\frac{i2\pi k j}{n}}|j\rangle,
\end{equation}
where $k = 1, ..., n$, with eigenvalues dependent on the interaction strength $\alpha$,
\begin{equation}
    \label{eq:eigenvalue_definition}
    \lambda_k(\alpha) = 2\sum_{j=1}^{n-1}\frac{1}{j^\alpha}\cos(\tfrac{2\pi k j}{n}).
\end{equation}
The largest eigenvalue is clearly at $k=n$, and the smallest is for $k=\frac{n}{2}$. We can also see that the eigenvalues are symmetric around $k=\frac{n}{2}$ such that $\lambda_{n-k}(\alpha) = \lambda_{k}(\alpha)$. The eigenvalues can be rewritten in terms of infinite sums, 
\begin{equation}
    \lambda_k(\alpha) = \sum_{j=1}^{\infty}\frac{e^{\frac{i2\pi k j}{n}}}{j^\alpha} + \sum_{j=1}^{\infty}\frac{e^{-\frac{i2\pi k j}{n}}}{j^\alpha}  -   \sum_{j=0}^{\infty}\frac{e^{\frac{i2\pi k j}{n}}}{(j+n)^\alpha} - \sum_{j=0}^{\infty}\frac{e^{-\frac{i2\pi k j}{n}}}{(j+n)^\alpha}.
\end{equation}
This means we can directly express the eigenvalues in terms of well-studied mathematical functions 
\begin{equation}
    \lambda_k(\alpha) =  \textrm{Li}_\alpha(e^{\frac{i2\pi k}{n}}) + \textrm{Li}_\alpha(e^{-\frac{i2\pi k}{n}}) -   \Phi(e^{\frac{i2\pi k}{n}}, \alpha, n) - \Phi(e^{-\frac{i2\pi k}{n}}, \alpha, n),
\end{equation}
where $\textrm{Li}_\alpha(z)$ is the polylogarithm~\cite[Section~25.12]{NIST:DLMF} and $\Phi(z, \alpha, n)$ is the Lerch transcendent function~\cite[Section~25.14]{NIST:DLMF}. In the case that $\alpha=1$ the eigenvalues are simplified because $\textrm{Li}_1(z) = -\ln (1-z)$.

The spectral condition, as introduced in the main text, and analytical predictions for fidelity and time~\cite{Chakraborty2020OptimalityWalks} apply for a normalised spectrum. We therefore rescale the eigenvalues such that the smallest eigenvalue is greater than $0$ and the largest eigenvalue is equal to $1$,
\begin{equation}
    \label{eq:rescaled_eigenvalue}
    \Tilde{\lambda}_k(\alpha) =  \frac{\lambda_k(\alpha) - \lambda_\frac{n}{2}(\alpha)}{\lambda_n(\alpha) - \lambda_\frac{n}{2}(\alpha)}.    
\end{equation}

For the analysis of the optimal search regime, we require the spectral gap and the parameters $S_q$. The spectral gap, $\Delta$, is the difference between the two highest eigenvalues. Here, it is a function of the interaction strength $\alpha$,
\begin{equation}
    \label{eq:spectral_gap}
    \Delta(\alpha) =  1 - \Tilde{\lambda}_{n-1}(\alpha).  
\end{equation}
The parameters $S_q$, in our case, are defined as \begin{equation}
    \label{eq:s_m_definition}
    S_q(\alpha) =  \frac{1}{n}\sum_{i=1}^{n-1}\frac{1}{(1-\Tilde{\lambda}_i(\alpha))^q},
\end{equation}
for integer $q \ge 1$. They were found to be important to the optimality of spatial search by Chakraborty et al.~\cite{Chakraborty2020OptimalityWalks} but do not appear to have a physical interpretation. In our case, $S_q$ depends only on the spectrum of $H$, the graph Hamiltonian.

First, we find an expression for the largest eigenvalue using Eq.~\eqref{eq:eigenvalue_definition}, 
\begin{align}
        \lambda_n(\alpha) &= 2 \sum_{j=1}^{n-1}\frac{1}{j^{\alpha}} \\
        &= 2\zeta(\alpha)  - 2\zeta(\alpha, n),
\end{align}
where we use the Hurwitz zeta function~\cite[Section~25.11]{NIST:DLMF},
\begin{equation}
    \zeta(\alpha,n) = \sum_{j=0}^\infty \frac{1}{(j+n)^\alpha},
\end{equation} 
and the Riemann zeta function~\cite[Section~25.2]{NIST:DLMF}, $\zeta(\alpha) = \zeta(\alpha, 1)$. We can then find an expression for the smallest eigenvalue, assuming $n$ is even, 
\begin{align}
        \lambda_\frac{n}{2}(\alpha) &= 2 \sum_{j=1}^{n-1}\frac{(-1)^j}{j^{\alpha}} \\
        &= 2 \left[\sum_{j=1}^{\frac{n}{2}-1}\frac{1}{(2j)^{\alpha}} - \sum_{j=0}^{\frac{n}{2}-1}\frac{1}{(2j+1)^{\alpha}} \right] \\
        &= 2^{1-\alpha}\left[2\zeta(\alpha)  - \zeta(\alpha, \tfrac{n}{2})-  2^\alpha\zeta(\alpha) + \zeta(\alpha, \tfrac{n+1}{2}) \right],
\end{align}
using $\zeta(\alpha,\frac{1}{2}) = (2^\alpha-1)\zeta(\alpha)$. For odd $n$, the smallest eigenvalues are $\lambda_\frac{n+1}{2}(\alpha)$ and $\lambda_\frac{n-1}{2}(\alpha)$. The corresponding expression tends to the same result for large $n$ and the following results apply asymptotically for both even and odd $n$. Asymptotically, $\zeta(\alpha,\frac{n}{2}+\frac{1}{2}) \sim \zeta(\alpha,\frac{n}{2})$, and we also have the series expansion~\cite{Nemes2017ErrorFunction}
\begin{equation}
    \zeta(\alpha,n) \sim \frac{n^{-\alpha}}{2} + \frac{n^{1-\alpha}}{\alpha-1} + \mathcal{O}(n^{-1-\alpha}),
\end{equation}
which is valid for $\alpha \ne 1$.

These asymptotic results lead to a simple expression for the difference between the largest eigenvalue and the smallest eigenvalue to leading orders in $n$,
\begin{equation}
    \lambda_n(\alpha) - \lambda_\frac{n}{2}(\alpha) \sim f(\alpha) - \frac{2n^{1-\alpha}}{\alpha-1} - n^{-\alpha},
\end{equation}
where $f(\alpha) =  4\zeta(\alpha) -2^{2-\alpha}\zeta(\alpha)$. For the asymptotic expression of the general eigenvalue $\lambda_k(\alpha)$, we require the series expansion 
\begin{equation}
    \label{eq:lerch_expansion}
    \Phi(z, \alpha, n) \sim \frac{n^{-\alpha}}{1-z} + \mathcal{O}(n^{-\alpha-1}),
\end{equation}
for $\textrm{Re}(z)<1$, and an expansion for the polylogarithm, for $k=1,2,\dots,n-1$ and $\alpha \ne 1,2,3,\dots$,
\begin{equation}
   \textrm{Li}_\alpha(e^{\frac{i2\pi k}{n}}) = \Gamma(1-\alpha)(-\tfrac{i2\pi k}{n})^{\alpha-1} + \sum_{j=0}^\infty \frac{\zeta(\alpha-j)}{j!}(\tfrac{i2\pi k}{n})^j,
\end{equation}
where $\Gamma(1-\alpha)$ is the standard gamma function. We therefore find 
\begin{equation}
    \label{eq:polylog_sum}
    \mathrm{Li}_\alpha(e^{\frac{i2\pi k}{n}}) + \mathrm{Li}_\alpha(e^{-\frac{i2\pi k}{n}}) = 2\zeta(\alpha) - f(\alpha)h(\alpha,\tfrac{n}{k}),
\end{equation}
where 
\begin{equation}
    \label{eq:h_definition}
    h(\alpha,\tfrac{n}{k}) = \frac{g_0(\alpha)}{f(\alpha)}\left(\frac{n}{k}\right)^{1-\alpha} + \sum_{m=1}^{\infty}\frac{g_m(\alpha)}{f(\alpha)}\left(\frac{n}{k}\right)^{-2m}
\end{equation}
with 
\begin{equation}
    \label{eq:g0_definition}
    g_0(\alpha)  = -2^{\alpha}\pi^{\alpha-1}\sin(\tfrac{\alpha\pi}{2})\Gamma(1-\alpha),
\end{equation}
and 
\begin{equation}
    \label{eq:gm_definition}
    g_m(\alpha)  = \frac{-2\zeta(\alpha - 2m)(2\pi i)^{2m}}{(2m)!}.
\end{equation}
This expansion is only valid up to $k=n/2$. However, this polylogarithm expression is symmetric about $k=n/2$, so that is not a problem.

From Eq.~\eqref{eq:lerch_expansion}, we can see that for large $n$ the Lerch transcendent function $\Phi(z, \alpha, n)$  will not be a leading order term in the rescaled eigenvalue of Eq.~\eqref{eq:rescaled_eigenvalue}. In the asymptotic limit, it can therefore be neglected, giving
\begin{align}
    \lambda_k(\alpha) - \lambda_{\frac{n}{2}}(\alpha) &\sim f(\alpha) - f(\alpha)h(\alpha, \tfrac{n}{k}) - \left(\frac{1}{1-e^{\frac{i2\pi k}{n}}} +  \frac{1}{1-e^{-\frac{i2\pi k}{n}}}\right)n^{-\alpha} + \left(\frac{1}{1-e^{i\pi}} +  \frac{1}{1-e^{-i\pi}}\right)n^{-\alpha} \\
    &= f(\alpha) - f(\alpha)h(\alpha, \tfrac{n}{k})
\end{align}

Finally, this provides an asymptotic expression for the rescaled eigenvalues 
\begin{equation}
    \label{eq:asymptotic_rescaled_eigenvalue}
    \Tilde{\lambda}_k(\alpha) \sim  \frac{1 - h(\alpha,\tfrac{n}{k})}{1 - \frac{2}{f(\alpha)} \frac{n^{1-\alpha}}{\alpha-1}}.
\end{equation}   
In the case $\frac{2}{f(\alpha)}\frac{n^{1-\alpha}}{\alpha-1} \ll 1$ – which is true asymptotically for $\alpha > 1$, but for $\alpha$ close to 1, only for very large $n$ – we can use the asymptotic scaling,
\begin{equation}
    \label{eq:rescaled_eigenvalue_expansion}
    \Tilde{\lambda}_k(\alpha) \sim  1 - h(\alpha,\tfrac{n}{k}) + (1-h(\alpha,\tfrac{n}{k})) \frac{2}{f(\alpha)}\frac{n^{1-\alpha}}{\alpha-1} +  \mathcal{O}(n^{-1-\alpha}). 
\end{equation}
The accuracy of the expansion is determined by the order of the approximation of $h(\alpha, \frac{n}{k})$. We have also further restricted our region of interest to $\alpha<3$. This can be justified retrospectively: once we have found the optimal search region for $0<\alpha<3$, it becomes clear that optimal search will not be possible for $\alpha>3$.  We cannot immediately take the first order term of the expansion for $h(\alpha,\frac{n}{k})$ because $k$ can be equal to $\frac{n}{2}$, which means we would introduce a finite residual. We investigate this question when considering the asymptotic scaling of the analytical amplitude of spatial search, later in the Supplementary Material.

\section{\label{sec:asymptotic_gap}Asymptotic scaling of the spectral gap}
The scaling of the spectral gap can be found accurately from Eq.~\eqref{eq:asymptotic_rescaled_eigenvalue} because only one eigenvalue is considered, $\Tilde{\lambda}_{n-1}(\alpha)$, where, by symmetry of the eigenvalues, we have $ \Tilde{\lambda}_{n-1}(\alpha) = \Tilde{\lambda}_{1}(\alpha)$. For $k=1$, asymptotically with $n$, the expansion of $h(\alpha,\tfrac{n}{k})$ is the first term,
\begin{equation}
     h(\alpha, n) \sim \frac{g_0(\alpha)}{f(\alpha)}n^{1-\alpha},
\end{equation}
which is valid for $\alpha<3$. Using the definition of the spectral gap in Eq.~\eqref{eq:spectral_gap}, we find the asymptotic scaling 
\begin{equation}
    \Delta(\alpha) \sim 1 -  \frac{1 - \frac{g_0(\alpha)}{f(\alpha)}n^{1-\alpha}}{1 - \frac{2}{f(\alpha)} \frac{n^{1-\alpha}}{\alpha-1}}.
\end{equation}
First, we consider the spectral gap for $\alpha<1$. In this case, both the numerator and denominator diverge with order $n^{1-\alpha}$, which dominates the constant 1 at large $n$. Thus, we find
\begin{align}
    \Delta(\alpha) &\sim 1 - \frac{(\alpha - 1)g_0(\alpha)}{2} \\
    &\sim 1 - 2^{\alpha-1}\pi^{\alpha-1}(1-\alpha) \sin(\tfrac{\alpha \pi}{2}) \Gamma(1-\alpha).
\end{align}
For $\alpha<1$, the spectral gap is therefore 
\begin{equation}
    \Delta(\alpha) = \mathcal{O}(1).
\end{equation}
For $\alpha=0$, the spectral gap tends to 1, as we would expect for the complete graph, and for $0<\alpha<1$ the spectral gap tends to a finite value greater than 0 and less than 1.

For $1<\alpha<3$, as used in the expansion of Eq.~\eqref{eq:rescaled_eigenvalue_expansion}, where $\frac{2}{f(\alpha)}\frac{n^{1-\alpha}}{\alpha-1} \ll 1$, we have an expansion for the spectral gap
\begin{equation}
    \Delta(\alpha) \sim \left(\frac{g_0(\alpha)}{f(\alpha)} - \frac{2}{(1-\alpha)f(\alpha)}\right) n^{1-\alpha},
\end{equation}
and the spectral gap is therefore 
\begin{equation}
    \label{eq:gap_scaling}
    \Delta(\alpha) = \mathcal{O}(n^{1-\alpha}).
\end{equation} For the case of $\alpha > 3$, the most significant scaling term of $\Delta(\alpha)$ becomes $\mathcal{O}(n^{-2})$. However, we are not concerned with this region as we do not find optimal spatial search.

\section{Asymptotic scaling of the analytical amplitude and time}
\label{sec:asymptotic_scaling_nu}
The analytical amplitude, as defined in Ref.~\cite{Chakraborty2020OptimalityWalks}, is given by 
\begin{equation}
    \label{eq:analytical_amplitude}
    \nu(\alpha) = \frac{S_1(\alpha)}{\sqrt{S_2(\alpha)}},
\end{equation}
with $S_q$ parameters defined in Eq.~\eqref{eq:s_m_definition}. From this, we can also define
\begin{equation}
    \label{eq:analytical_fidelity}
    F_\infty(\alpha) = \nu(\alpha)^2 = \frac{S_1(\alpha)^2}{S_2(\alpha)},
\end{equation}
which is the large $n$ analytical prediction for fidelity used in Fig.~\ref{fig:timefidelity} in the main text. In order to find the asymptotic scaling of the analytical amplitude, the asymptotic scaling of $S_1$ and $S_2$ must be found. In this section the primary aim is to find the $\alpha$ values for which $\nu(\alpha)$ converges to a non-zero value, and thus $F_\infty(\alpha)$ converges to a non-zero value. From the arguments given the main text, we know that $F_\infty$ must converge to a non-zero value for $\alpha<1$. It is therefore sufficient to only consider the region $\alpha>1$ in this section.

The asymptotic scaling for $S_q$ with respect to $n$ is more complicated than in the case of the spectral gap because it requires $\Tilde{\lambda}_k(\alpha)$ for every $k$. From Eq.~\eqref{eq:s_m_definition}, we have 
\begin{align}
    S_q(\alpha) &\sim \frac{1}{n} \sum_{k=1}^{n-1} \left( h(\alpha, \tfrac{n}{k}) - \tfrac{2n^{1-\alpha}}{f(\alpha)(\alpha-1) }\right)^{-q}  \\
    & \sim \left(\frac{1}{n} + \frac{2qn^{-\alpha}}{f(\alpha)(\alpha-1)} \right) \sum_{k=1}^{\frac{n}{2}}\frac{2}{h(\alpha, \tfrac{n}{k})^{q}} \\
    & \sim \frac{2}{n} \sum_{k=1}^{\frac{n}{2}}\frac{1}{h(\alpha, \tfrac{n}{k})^{q}}
\end{align}
where we are only interested in the dominant order of $n$ for our case of $\alpha>1$. Explicitly, we have
\begin{equation}
    S_q(\alpha) \sim \frac{2}{n} \sum_{k=1}^{\frac{n}{2}} \left( \frac{g_0(\alpha)}{f(\alpha)}\left(\frac{n}{k}\right)^{1-\alpha} + \sum_{m=1}^{\infty}\frac{g_m(\alpha)}{f(\alpha)}\left(\frac{n}{k}\right)^{-2m} \right)^{-q},
\end{equation}
which can be expanded using the binomial theorem as 
\begin{align}
    S_q(\alpha) &\sim \frac{2}{n} \sum_{k=1}^{\frac{n}{2}} \sum_{j=0}^{\infty} \binom{-q}{j} A^{-q-j} B^{j}, \\
    &= \frac{2}{n} \Biggl[ \sum_{k=1}^{\frac{n}{2}} A^{-q} + \sum_{k=1}^{\frac{n}{2}}\sum_{j=1}^{\infty}\binom{-q}{j}A^{-q-j} B^{j} \Biggr],
\end{align}
where 
\begin{equation}
    A = \frac{g_0(\alpha)}{f(\alpha)}\left(\frac{n}{k}\right)^{1-\alpha}
\end{equation}
and 
\begin{equation}
    B = \sum_{m=1}^{\infty}\frac{g_m(\alpha)}{f(\alpha)}\left(\frac{n}{k}\right)^{-2m}.
\end{equation}
It is clear that the largest value of $B$ occurs when $k=\frac{n}{2}$ and the smallest value of $B$ when $k=1$. We can therefore bound $B$ from both sides. From Eqs.~13 and 14, we have 
\begin{equation}
    B = \frac{2\zeta(\alpha)- \mathrm{Li}_\alpha(e^{\frac{i2\pi k}{n}}) - \mathrm{Li}_\alpha(e^{-\frac{i2\pi k}{n}})}{f(\alpha)}-A,
\end{equation}
where $f(\alpha) =  4\zeta(\alpha) -2^{2-\alpha}\zeta(\alpha)$, and
\begin{equation}
    A = \frac{g_0(\alpha)}{f(\alpha)}\left(\frac{n}{k}\right)^{1-\alpha}.
\end{equation}
For $k=\frac{n}{2}$, we use $\textrm{Li}_\alpha(-1) = -(1-2^{1-\alpha})\zeta(\alpha)$ to find
\begin{equation}
    B_{\textrm{max}} = 1 - A.
\end{equation}
As $n$ tends to infinity, we can consider the minimum value of $B$, which is for $k=1$, giving
\begin{equation}
    B_{\textrm{min}} = - A = 0,
\end{equation}
using $\mathrm{Li}_\alpha(e^{\frac{i2\pi k}{n}}) \rightarrow \zeta(\alpha)$, as $n$ goes to infinity. It is then straightforward to numerically verify that $0<|B|<1 $ for $\alpha$ in the range 
\begin{equation}
    1+\varepsilon< \alpha< 2-\varepsilon
\end{equation}
for some small constant $\varepsilon$. This allows us to investigate the asymptotic behaviour of $S_q(\alpha) \sim F + C$, where we have defined the first order term 
\begin{equation}
    F = \frac{2}{n} \sum_{k=1}^{\frac{n}{2}} A^{-q},
\end{equation}
and a higher order correction term 
\begin{equation}
    C = \frac{2}{n} \sum_{k=1}^{\frac{n}{2}}\sum_{j=1}^{\infty}\binom{-q}{j}A^{-q-j} B^{j}.
\end{equation}
We now use results from lemmas that are proved in the following section of the Supplementary Material. From Lemma~\ref{lemma:convergence_of_F}, we find that $F$ converges for $\alpha < 1 + \frac{1}{q}$ and, using Lemma~\ref{lemma:convergence_of_C}, we find that $C$ converges for $\alpha < \mathrm{min}(2-\varepsilon, \frac{q+4}{q+1})$. Overall, we see that $S_1$ converges for $\alpha < 2$, and $S_2$ converges for $\alpha < 1.5$. For the region where $S_q$ does not converge, and in the case that $F$ diverges and $C$ converges, the most significant order of $n$ would come from $F$. We can therefore predict 
\begin{equation}
    \label{eq:s_q_divergences_rate}
    S_q(\alpha) = \mathcal{O}(n^{q\alpha-q-1})
\end{equation}
for $\alpha$ just above $1 + \frac{1}{q}$, where $F$ begins to diverge but where $C$ still converges. Eq.~\eqref{eq:s_q_divergences_rate} demonstrates the most significant order of $n$. However, we do not have an asymptotic limit because we do not have a strict value for $C$ when it converges. We can see approximately how large $C$ is compared with $F$ by looking at the highest-order term and constant in $F$ relative to numerical calculation for $S_q$. Fig.~\ref{fig:approximate_s1} demonstrates the size of $C$ for $S_1$. We have used an approximation of $S_q$,
\begin{equation}
    \bar{S}_q(\alpha) = \frac{2f(\alpha)^q}{g_0(\alpha)^q}\left(n^{q\alpha-q-1} \zeta(q\alpha -q) + \frac{2^{q\alpha -q-1}}{1+q-q\alpha}\right),
\end{equation}
where $\bar{S}_q(\alpha) \approx S_q$. We note that for $S_2$, $C$ is larger than for $S_1$, especially for the divergence region. 
\begin{figure}[]
    \centering
    \includegraphics[scale=0.48]{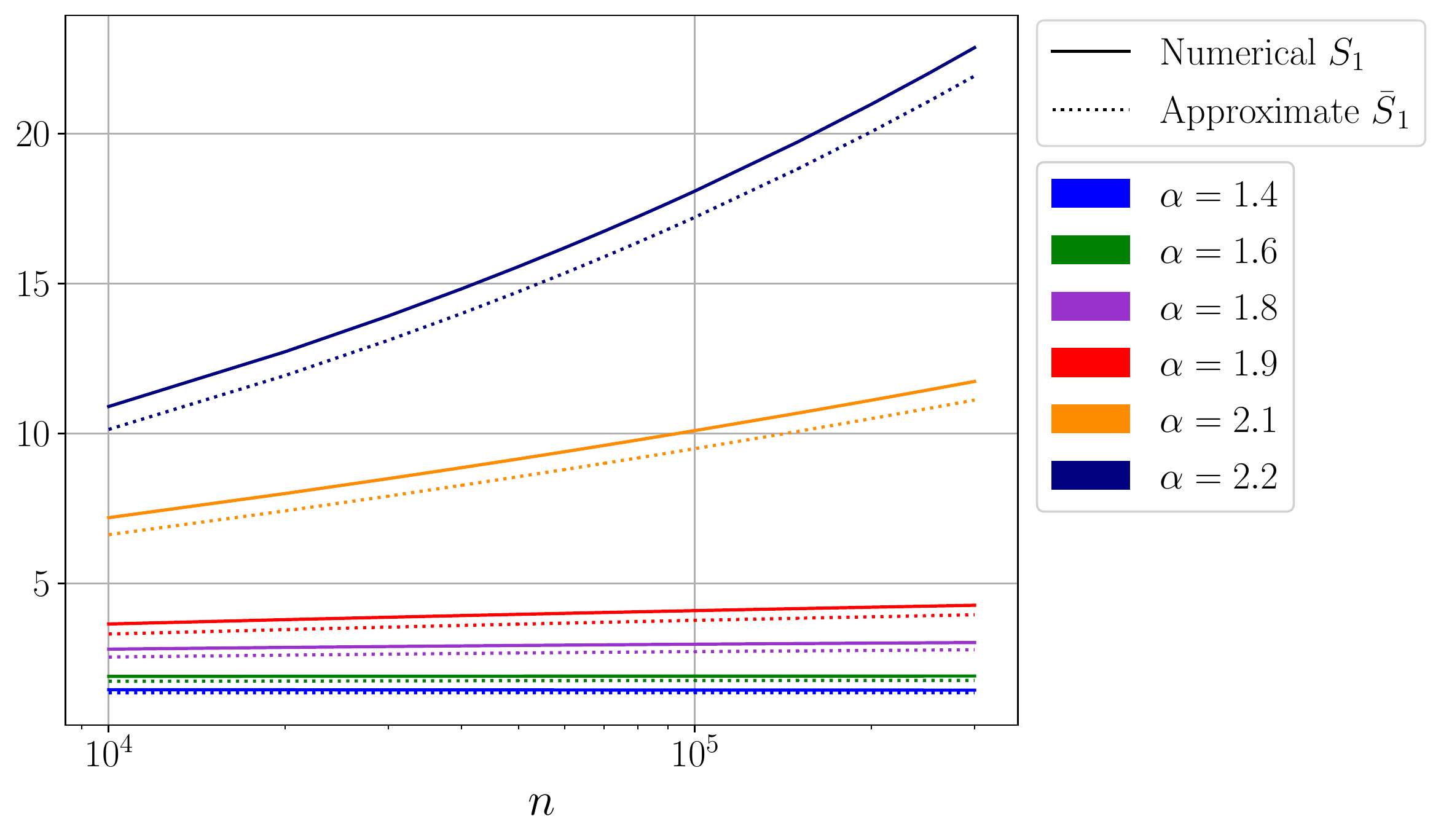}
    \caption{Approximate $\bar{S}_1$ compared to $S_1$ from numerical calculations of Eq.~\eqref{eq:s_m_definition}. The size of the higher order term $C$ is the difference between the numerical value $S_1$ and the approximate value $\bar{S}_1$. Divergence for $\alpha > 2$ is suggested, as predicted by Eq.\eqref{eq:s_q_divergences_rate}.}
    \label{fig:approximate_s1}
\end{figure}
\begin{figure}
    \centering
    \includegraphics[scale=0.48]{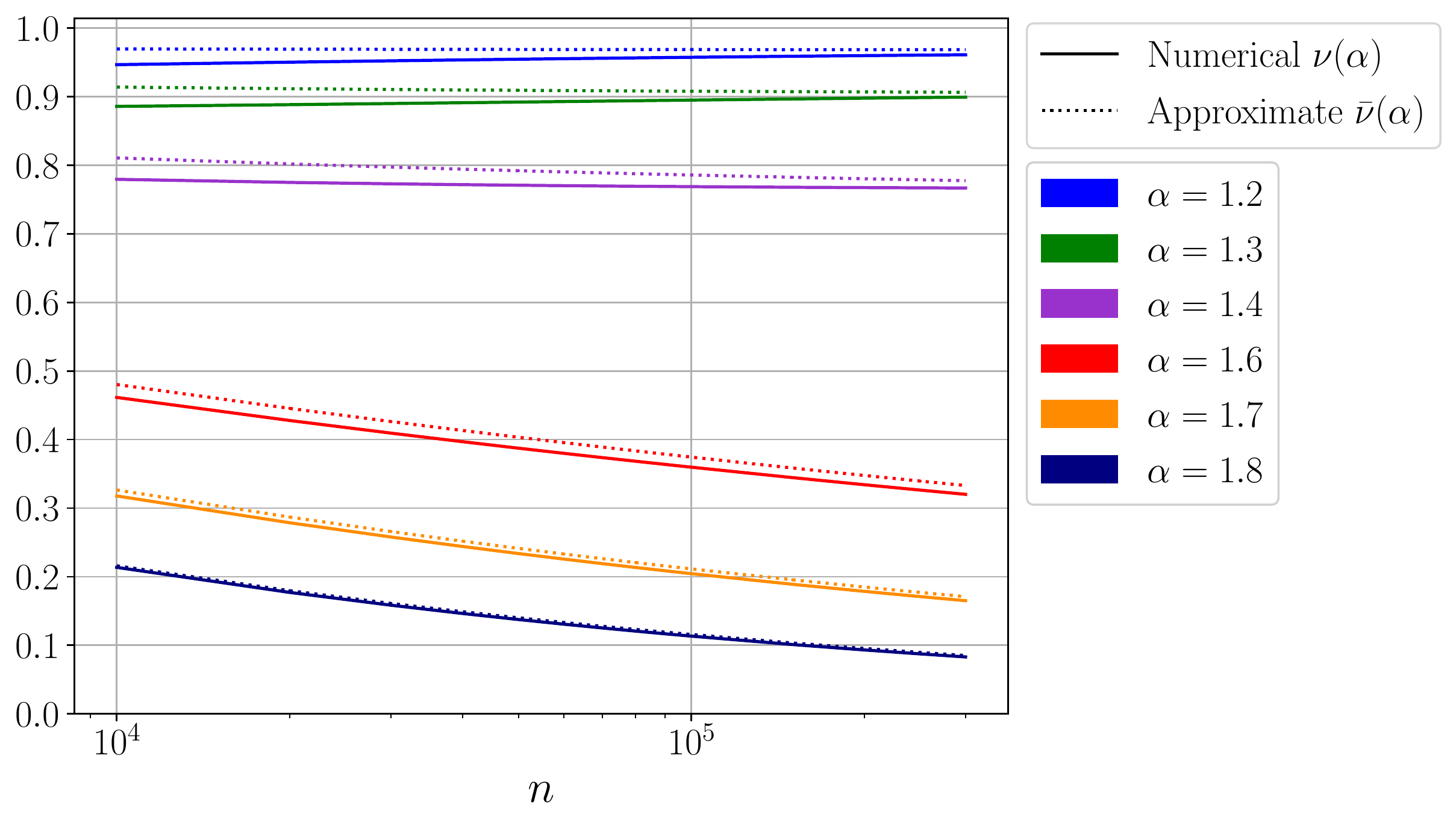}
    \caption{Approximate $\bar{\nu}(\alpha)$ compared to $\nu(\alpha)$ calculated using numerical calculations of $S_1$ and $S_2$ from Eq.~\eqref{eq:s_m_definition}. The difference in behaviour around the asymptotic transition $\alpha = 1.5$ is shown, with $\nu(\alpha)$ tending to a finite value for $\alpha<1.5$ and tending to 0 for $\alpha>1.5$.}
    \label{fig:approximate_nu}
\end{figure}

We can now find an approximation for the analytical amplitude 
\begin{align}
    \bar{\nu}(\alpha) &= \frac{\bar{S}_1(\alpha)}{\sqrt{\bar{S}_2(\alpha)}} \\
    &= \frac{\sqrt{2} \left(n^{\alpha-2} \zeta(\alpha -1) + \frac{2^{\alpha -2}}{2-\alpha}\right) }{\sqrt{n^{2\alpha-3} \zeta(2\alpha -2) + \frac{2^{2\alpha -3}}{3-2\alpha}}}, \label{eq:approximate_nu}
\end{align}
where $\bar{\nu}(\alpha) \approx \nu(\alpha)$. The accuracy of this approximation is illustrated in Fig.~\ref{fig:approximate_nu}. We can use Eq.~\eqref{eq:approximate_nu} to give an approximate value for the analytical amplitude $\nu(\alpha)$ asymptotically, in the case $1<\alpha<1.5$,
\begin{equation}
    \lim_{n\rightarrow\infty}\bar{\nu}(\alpha) = \frac{\sqrt{3-2\alpha}}{2-\alpha}.
\end{equation}
This approximation is used to find the asymptotic approximation of the fidelity, $F_\infty(\alpha) = \nu(\alpha)^2$, in the thermodynamic limit, 
\begin{equation}
    \label{eq:approx_asymptotic_fidelity}
    \bar{F}_\infty(\alpha) = \lim_{n\rightarrow\infty}\bar{\nu}(\alpha)^2,
\end{equation}
which is used in the inset of Fig.~\ref{fig:timefidelity}(a) in the main text.

\begin{figure}[ht]
    \centering
		\includegraphics[width=0.46\textwidth]{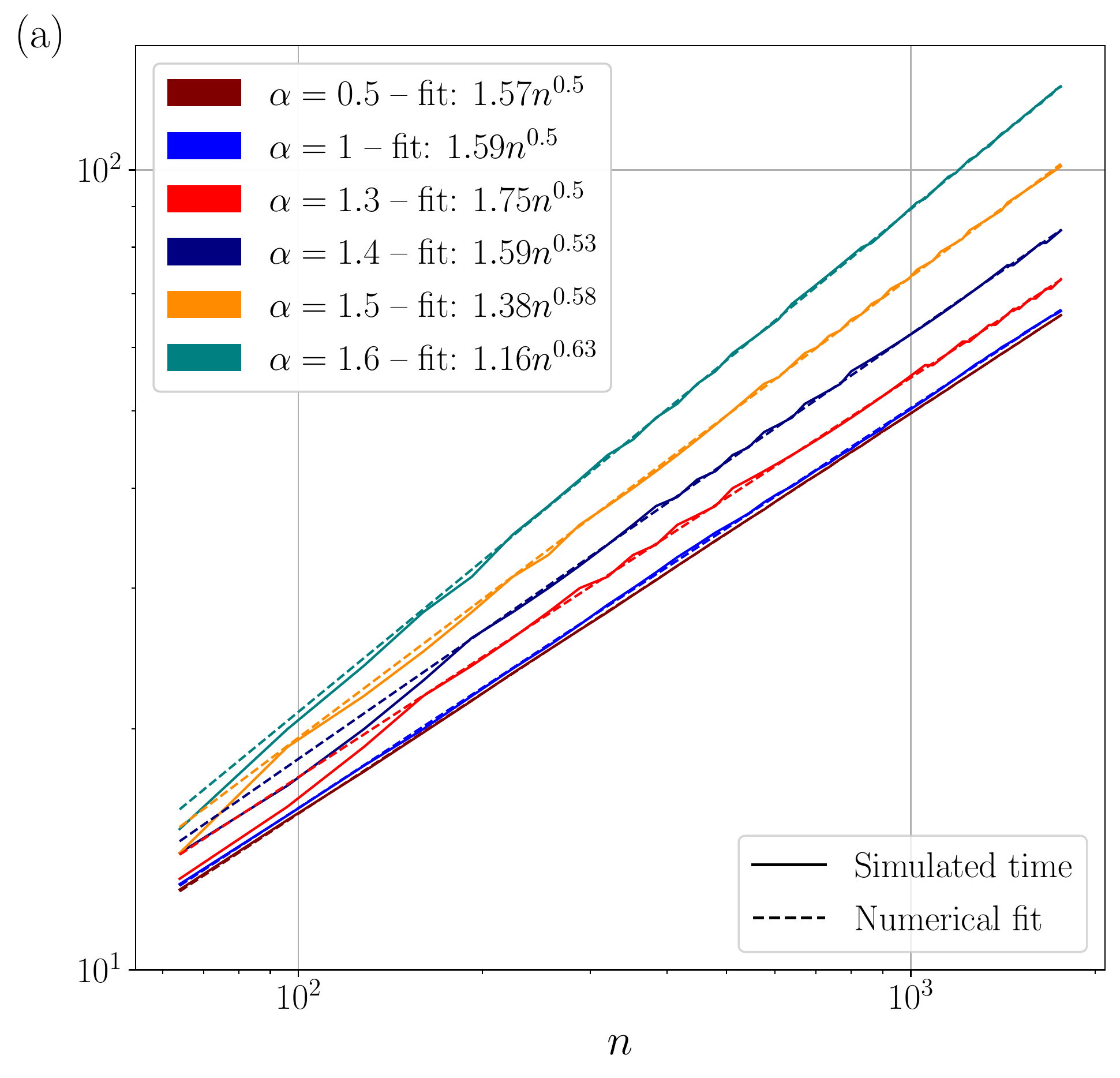}
        \includegraphics[width=0.46\textwidth]{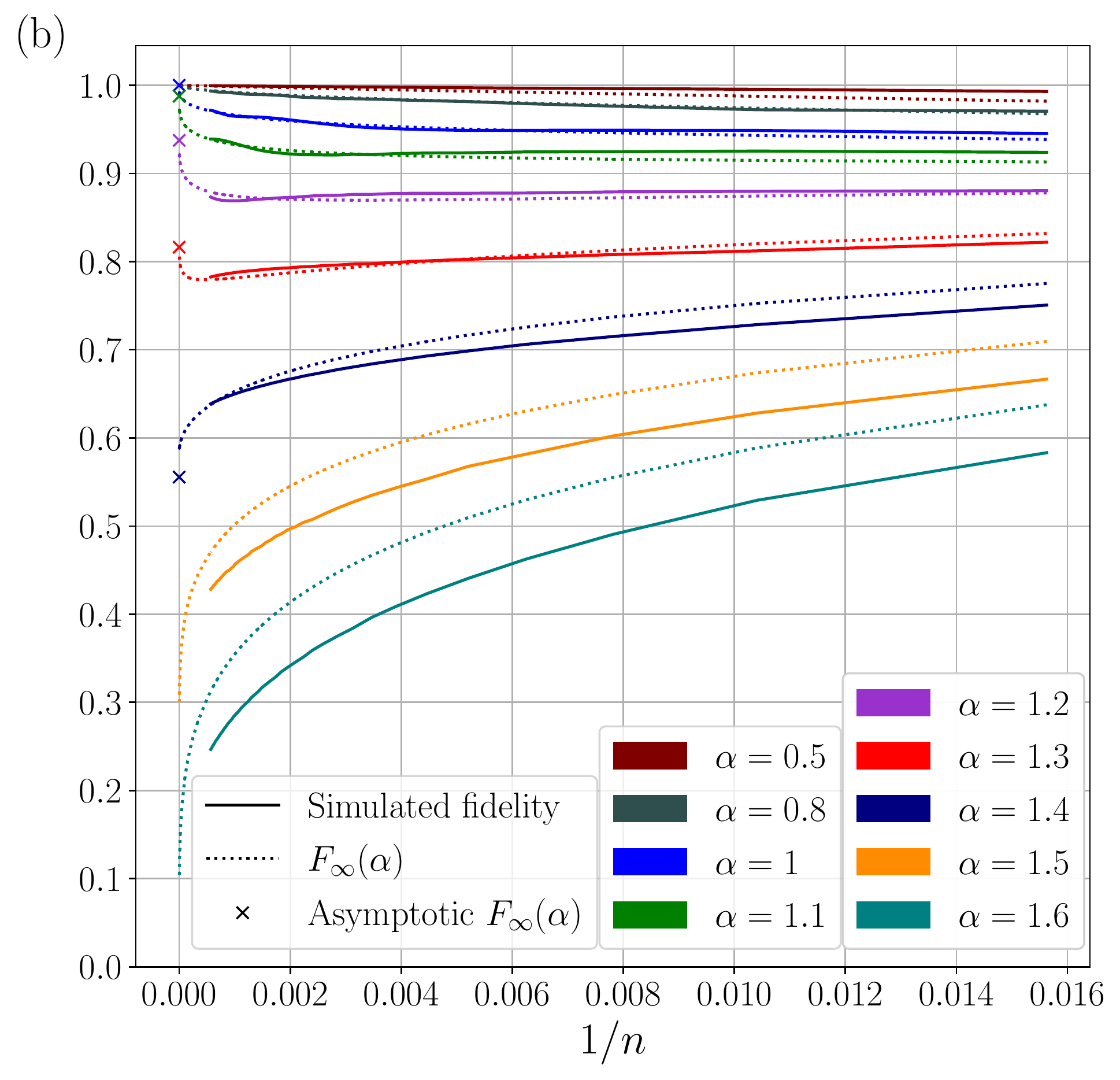}
		\caption{
		(a) Time to reach the maximum fidelity of the marked state against the number of spins $n$. The line (dashed) of best fit is also shown. 
	  (b) Maximum fidelity of the spatial search against $1/n$. Numerical results are compared with the analytical prediction of $F_\infty(\alpha)$ for the fidelity. $F_\infty(\alpha) = \nu(\alpha)^2$ is calculated using numerical calculations of $S_1$ and $S_2$. For $\alpha = 1.5$ and $\alpha=1.6$, we do not have optimal spatial search asymptotically. The asymptotic $F_\infty(\alpha)$ is an approximation given by Eq.~\eqref{eq:approx_asymptotic_fidelity}.}
	\label{fig:timefidelity_SM}
\end{figure}
Finally, we can now determine the region for which $\nu(\alpha)$ tends to a finite value. Due to $S_1$ converging for $\alpha < 2$ and $S_2$ converging for $\alpha < 1.5$, it follows that $\nu(\alpha)$ must converge for $\alpha < 1.5$. For $1.5 < \alpha < 2$, $S_1$ converges and $S_2$ diverges, therefore $\nu(\alpha) \rightarrow 0$. For $\alpha > 2$, $\sqrt{S_2}$ diverges faster than $S_1$, so $\nu(\alpha) \rightarrow 0$. Thus, the region for asymptotic convergence to non-zero analytical amplitude $\nu(\alpha)$, and therefore convergence to non-zero $F_\infty(\alpha)$, is $\alpha < 1.5$. 

The analytical time for optimal spatial search~\cite{Chakraborty2020OptimalityWalks} can be written in terms of the analytical fidelity $F_\infty(\alpha) = \nu(\alpha)^2$,
\begin{equation}
    T(\alpha) \approx \frac{\pi}{2}\sqrt{\frac{n}{F_\infty(\alpha)}}.
\end{equation}
Thus, the asymptotic scaling of time is related to the asymptotic scaling of $\nu(\alpha)$. As $\alpha$ increases the scaling of fidelity with $n$ becomes more significant as the order of $n$ contributing to $\tilde{\lambda}_k(\alpha)$ increases. It is therefore sufficient to show that the time scaling approaches optimal for large $n$ in the regime $1<\alpha<1.5$. In this regime, using the approximation for analytical amplitude in Eq.~\eqref{eq:approximate_nu} the analytical time is given by 
\begin{align}
    T &\approx \frac{\pi}{2}\sqrt{n}\sqrt{\frac{b_0 + b_1 n^{2\alpha-3}}{\left( a_0 + a_1 n^{\alpha-2} \right)^2}}   \\ 
    &= \frac{\pi}{2}\sqrt{n}\sqrt{a_0^{-2}\left( b_0 + b_1 n^{2\alpha-3}\right)\left( 1 - 2 \tfrac{a_1}{a_0} n^{\alpha-2} + \mathcal{O}(n^{2\alpha-4}) \right)} \\
    &= \frac{\pi}{2}\sqrt{\frac{n b_0}{a_0^2}}\sqrt{1 + \frac{b_1}{b_0} n^{2\alpha-3} - \frac{2 a_1}{a_0} n^{\alpha -2} - \frac{2 a_1 b_1}{a_0 b_0} n^{3\alpha -5} + \mathcal{O}(n^{2\alpha-4})},
\end{align}
where $a_0 = \sqrt{2}\frac{2^{\alpha -2}}{2-\alpha}$, $a_1 = \sqrt{2}\zeta(\alpha -1)$, $b_0 = \frac{2^{2\alpha -3}}{3-2\alpha}$, and $b_1 = \zeta(2\alpha - 2)$. The dominant term therefore becomes 
\begin{equation}
    T \approx \frac{\pi}{2}\sqrt{\frac{n}{\bar{F}_{\infty}(\alpha)}}
\end{equation}
for large $n$, where $\bar{F}_{\infty}(\alpha)$, is the asymptotic approximation of the fidelity. 
The accuracy of this asymptotic scaling for low $n$ can be verified directly by fitting the simulated time to reach the maximum fidelity for various $\alpha$, see Fig.~\ref{fig:timefidelity_SM}(a). For $\alpha \leq 1.3$, we find that $T = \mathcal{O}(\sqrt{n})$ is very accurate. For $1.3 <\alpha < 1.5$, the deviation is small up to the $n$ we have simulated. As $n$ becomes larger, the time scaling for this range will also approach $T = \mathcal{O}(\sqrt{n})$, as $F_{\infty}(\alpha)$ tends to a constant, see Fig.~\ref{fig:timefidelity_SM}(b).

\section{\label{sec:lemmas}Lemmas used for the convergence of \texorpdfstring{$\bm{S_q}$}{}}
\begin{lemma}
\label{lemma:convergence_of_F}
The first order term $F$, \[ F = \frac{2}{n} \sum_{k=1}^{\frac{n}{2}} A^{-q},\]
with $A = \frac{g_0(\alpha)}{f(\alpha)}\left(\frac{n}{k}\right)^{1-\alpha}$,
tends to a finite value as $n\rightarrow\infty$ for all integers $q \ge 1$ for $\alpha$ in the range 
\[1 < \alpha < 1 + \frac{1}{q}.\]
\end{lemma}
\begin{proof}
We can write the summation in the definition of $F$ in terms of the generalised harmonic numbers,
\begin{equation}
    H_{\frac{n}{2},s} = \sum_{j=1}^{\frac{n}{2}}\frac{1}{j^{s}},
\end{equation}
giving
\begin{align}
    \frac{2}{n} \sum_{k=1}^{\frac{n}{2}} A^{-q} &=  \frac{2g_0(\alpha)^{-q}}{f(\alpha)^{-q}}n^{q\alpha-q-1} \sum_{k=1}^{\frac{n}{2}}\frac{1}{k^{q\alpha-q}} \\
    &= \frac{2g_0(\alpha)^{-q}}{f(\alpha)^{-q}}n^{q\alpha-q-1} H_{\frac{n}{2},q\alpha-q}.
\end{align}
Using the Euler-Maclaurin formula~\cite[Section~2.10]{NIST:DLMF}, we have an expansion for $H_{\frac{n}{2},s}$ for large $n$ for $s \ne 1$,
\begin{equation}
    H_{\frac{n}{2},s} = \zeta(s) + \frac{1}{1-s}\left(\frac{n}{2}\right)^{1-s} + \frac{1}{2}\left(\frac{n}{2}\right)^{-s} - \frac{s}{12}\left(\frac{n}{2}\right)^{-s-1} + \mathcal{O}\left((\tfrac{n}{2})^{-s-2}\right).
\end{equation}
This expansion gives 
\begin{equation}
    \frac{2}{n} \sum_{k=1}^{\frac{n}{2}} A^{-q} = \frac{2f(\alpha)^q}{g_0(\alpha)^q}\left(n^{q\alpha-q-1} \zeta(q\alpha -q) + \frac{2^{q\alpha -q-1}}{1+q-q\alpha} + \mathcal{O}(n^{-1})\right).
\end{equation}
We are only considering the most significant order of $n$ for the asymptotic behaviour. For convergence, the exponent of $n$ must be negative and we have $q\alpha - q - 1 < 0$, which gives the condition 
\begin{equation}
   1 < \alpha < 1 + \frac{1}{q}
\end{equation}
for $F$ to tend to a finite value asymptotically with $n$ for all integer $q \ge 1$.
\end{proof}

\begin{lemma}
\label{lemma:convergence_of_C}
Assuming $|B| < 1$, the correction term $C$, 
\[
    C = \frac{2}{n} \sum_{k=1}^{\frac{n}{2}}\sum_{j=1}^{\infty}\binom{-q}{j}A^{-q-j} B^{j},
\]
tends to a finite value as $n \rightarrow \infty$ for all integers $q \ge 1$ for $\alpha$ in the range 
\[1 < \alpha < \frac{q+4}{q+1},\]
\end{lemma}
\begin{proof}Since $|B|<1$, we know $|B|>|B^j|$ for all integers $j>1$. We can therefore bound each term of the summation and bound the entire sum:
\begin{align}
    C &< \frac{2}{n} \sum_{k=1}^{\frac{n}{2}}\sum_{j=1}^{\infty}\binom{-q}{j} A^{-q-j} B \\
    &= \frac{2}{n} \sum_{k=1}^{\frac{n}{2}}\sum_{j=1}^{\infty}\binom{-q}{j}\frac{g_0(\alpha)^{-q-j}}{f(\alpha)^{-q-j}} \sum_{m=1}^{\infty}\frac{g_m(\alpha)}{f(\alpha)}\left(\frac{n}{k}\right)^{q\alpha-q+j\alpha-j-2m} \\
    &= 2 \sum_{j=1}^{\infty}\binom{-q}{j}\frac{g_0(\alpha)^{-q-j}}{f(\alpha)^{-q-j}} \sum_{m=1}^{\infty}\frac{g_m(\alpha)}{f(\alpha)}n^{q\alpha-q+j\alpha-j-1-2m}H_{\frac{n}{2},q\alpha-q+j\alpha-j-2m},
\end{align}
where, without loss of generality, we have assumed $B$ is positive. We have also used the fact that we know $C$ converges for finite $n$ when the summations were swapped. The generalised harmonic number has been used,
\begin{equation}
    \label{eq:asymptotic_sk}
    H_{\frac{n}{2},s} = \sum_{j=1}^{\frac{n}{2}}\frac{1}{j^{s}},
\end{equation}
for which we have an expansion for large $n$ using the Euler-Maclaurin formula for $s \ne 1$,
\begin{equation}
    H_{\frac{n}{2},s} = \zeta(s) + \frac{1}{1-s}\left(\frac{n}{2}\right)^{1-s} + \frac{1}{2}\left(\frac{n}{2}\right)^{-s} - \frac{s}{12}\left(\frac{n}{2}\right)^{-s-1} + \mathcal{O}\left((\tfrac{n}{2})^{-s-2}\right).
\end{equation}
Substituting in the asymptotic expansion for the harmonic numbers, we find 
\begin{multline}
    C < 2 \sum_{j=1}^{\infty}\binom{-q}{j}\frac{g_0(\alpha)^{-q-j}}{f(\alpha)^{-q-j}} \sum_{m=1}^{\infty}\frac{g_m(\alpha)}{f(\alpha)}\Biggl( n^{q\alpha-q+j\alpha-j-1-2m}\zeta(q\alpha-q+j\alpha-j-2m) + \frac{2^{q\alpha-q+j\alpha-j-1-2m}}{2m+1-j\alpha+j+q-q\alpha} \\ + 2^{q\alpha-q+j\alpha-j-1-2m}n^{-1} + \mathcal{O}(n^{-2}) \Biggr).
\end{multline}
The terms of $\mathcal{O}(n^{-1})$ tend to 0 asymptotically for $n$. We can also see that if $q\alpha-q+j\alpha-j-1-2m > 0$, then as $n$ tends to infinity the terms in the sum also tend to infinity and the sum diverges. In order for the possibility of a convergent sum, we therefore require
\begin{equation}
    \alpha < \frac{q+j+1+2m}{q+j},
\end{equation}
which must be the case for all values of $j \ge 1$ and $m \ge 1$. This gives 
\begin{equation}
    \alpha < \frac{q+4}{q+1}.
\end{equation}
For these values of $\alpha$, we have 
\begin{align}
    C &< 2 \sum_{j=1}^{\infty}\binom{-q}{j}\frac{g_0(\alpha)^{-q-j}}{f(\alpha)^{-q-j}} \sum_{m=1}^{\infty}\frac{g_m(\alpha)}{f(\alpha)}\frac{2^{q\alpha-q+j\alpha-j-1-2m}}{2m+1-j\alpha+j+q-q\alpha} \\
    &= 2 \sum_{j=1}^{\infty}\binom{-q}{j}\frac{g_0(\alpha)^{-q-j}}{f(\alpha)^{-q-j}} \sum_{m=1}^{\infty}\frac{g_m(\alpha)}{f(\alpha)}2^{-2m}\left(\frac{2^{q\alpha-q+j\alpha-j-1}}{2m+1-j\alpha+j+q-q\alpha}\right).
\end{align}
We can introduce a positive constant $c \ge 1$, such that 
\begin{equation}
\frac{2^{q\alpha-q+j\alpha-j-1}}{c(2m+1-j\alpha+j+q-q\alpha)} < 1,
\end{equation}
since we already have $2m+1-j\alpha+j+q-q\alpha > 0$.
This gives 
\begin{equation}
    C < 2c  \sum_{j=1}^{\infty}\binom{-q}{j}\frac{g_0(\alpha)^{-q-j}}{f(\alpha)^{-q-j+1}} \sum_{m=1}^{\infty}g_m(\alpha)2^{-2m}.
\end{equation}
From Eqs.~\eqref{eq:polylog_sum} and~\eqref{eq:h_definition}, we have 
\begin{equation}
    \sum_{m=1}^{\infty}g_m(\alpha)2^{-2m} = 4\zeta(\alpha) - 2^{2-\alpha}\zeta(\alpha) - 2^{1-\alpha}g_0(\alpha),
\end{equation}
where we have set $k=\frac{n}{2}$, and used $\textrm{Li}_\alpha(-1) = -(1-2^{1-\alpha})\zeta(\alpha)$. Thus, we find 
\begin{equation}
    C < 2c (4\zeta(\alpha) - 2^{2-\alpha}\zeta(\alpha) - 2^{1-\alpha}g_0(\alpha)) \frac{f(\alpha)^{q+1}}{g_0(\alpha)^q}  \sum_{j=1}^{\infty}\binom{-q}{j}\left( \frac{f(\alpha)}{g_0(\alpha)}\right)^j.
\end{equation}
Using Lemma~\ref{lemma:ratio_bound}, we have $|\frac{f(\alpha)}{g_0(\alpha)}| < 1$ for $1<\alpha<3$. We can therefore use the binomial theorem again to show
\begin{equation}
    \sum_{j=1}^{\infty}\binom{-q}{j}\left( \frac{f(\alpha)}{g_0(\alpha)}\right)^j = \left(1 + \frac{f(\alpha)}{g_0(\alpha)} \right)^{-q} - 1,
\end{equation}
which is of course a finite value. We have proved that $C$ tends to a finite value as $n \rightarrow \infty$ for all $q \ge 1$ for $\alpha$ in the range 
\begin{equation}
    1 < \alpha < \frac{q+4}{q+1}.
\end{equation}
\end{proof}

\begin{lemma}
\label{lemma:ratio_bound} $|\frac{f(\alpha)}{g_0(\alpha)}| < 1$ for $\alpha$ in the region $1 < \alpha < 3$, where $f(\alpha) = 4\zeta(\alpha) - 2^{2-\alpha}\zeta(\alpha)$ and $g_0(\alpha) = -2^\alpha \pi^{\alpha-1} \sin(\frac{\alpha\pi}{2})\Gamma(1-\alpha)$.
\end{lemma}
\begin{proof}
We can use the Riemann functional equation~\cite[Section~25.4]{NIST:DLMF},
\begin{equation}
    \zeta(\alpha) = 2^\alpha \pi^{\alpha-1} \sin(\tfrac{\alpha\pi}{2})\Gamma(1-\alpha) \zeta(1-\alpha),
\end{equation}
to give 
\begin{equation}
    f(\alpha) = (4-2^{2-\alpha})2^\alpha \pi^{\alpha-1} \sin(\tfrac{\alpha\pi}{2})\Gamma(1-\alpha) \zeta(1-\alpha).
\end{equation}
We now have 
\begin{equation}
    \frac{f(\alpha)}{g_0(\alpha)} = -(4-2^{2-\alpha})\zeta(1-\alpha).
\end{equation}
$\zeta(0)=-\frac{1}{2}$, $\zeta(-1)=-\frac{1}{12}$, and $\zeta(-2) = 0$, so, since we know there are no poles, and there are no turning points, we find that $|\frac{f(\alpha)}{g_0(\alpha)}| < 1$ for $1 < \alpha < 3$ – where we have used Ref.~\cite{Lee2014Zeros1/2}, which states that the derivative of $\zeta(s)$ does not have a zero for $-2<\textrm{Re}(s)<0$ meaning there cannot be a turning point.
\end{proof}

\twocolumngrid

\end{document}